\documentclass[letterpaper,11pt]{article} 

\usepackage{amsmath, amsfonts, amssymb, amsthm, amscd, graphicx, 
enumerate, footmisc, mathtools, xcolor}
\providecommand{\U}[1]{\protect\rule{.1in}{.1in}}
\def\theenumi{\arabic{enumi}}

\def\theenumii{\alph{enumii}}
\def\p@enumii{\theenumi.}

\def\theenumiii{\arabic{enumiii}}
\def\p@enumiii{(\theenumi)(\theenumii)}

\def\p@enumiv{\p@enumiii.\theenumiii}
\newcommand{\bbR}{\mathbb R}
\newcommand{\bbC}{\mathbb C}

\newcommand{\diag}{\operatorname{diag}}

\newcommand{\calP}{\mathcal P}

\usepackage{fullpage}
\usepackage{mathrsfs}
\newtheorem{theorem}{Theorem}[section]

\newtheorem{definition}[theorem]{Definition}

\newtheorem{lemma}[theorem]{Lemma}

\numberwithin{equation}{section}

\usepackage{titling}
\usepackage{enumitem} 
\setlist{noitemsep,topsep=0pt,parsep=0pt} 
\usepackage{subcaption}
\usepackage{tikz}

\numberwithin{equation}{section}
\usepackage[bookmarks=true,hypertexnames=false]{hyperref}
\hypersetup{colorlinks=true, citecolor=blue, linkcolor=red, urlcolor=blue}
\makeatletter

\newcommand{\Rmnum}[1]{\expandafter\@slowromancap\romannumeral #1@}
\makeatother

\newenvironment{remark}{\medskip{\bfseries \noindent Remark:}}{\par\medskip}{\par\medskip}

\newcommand{\EVAL}{\operatorname{EVAL}}

\newcommand{\simp}{\operatorname{simp}}

\usepackage[export]{adjustbox} 
\newcommand{\llbracket}{[\![}
\newcommand{\rrbracket}{]\!]}

\begin{document}

\title{\textbf{A dichotomy for bounded degree graph homomorphisms with nonnegative weights}}

\vspace{0.3in}

\author{Artem Govorov\thanks{Department of Computer Sciences, University of Wisconsin-Madison. Supported by NSF CCF-1714275. 
 } \thanks{Artem Govorov is the author's preferred spelling
of his name,
rather than the official spelling Artsiom Hovarau.
 } \\ {\tt hovarau@cs.wisc.edu}
\and
Jin-Yi Cai\thanks{Department of Computer Sciences, University of Wisconsin-Madison. Supported by NSF CCF-1714275.} \\ {\tt jyc@cs.wisc.edu} 
\and Martin Dyer\thanks{School of Computing, University of Leeds. Supported by EPSRC grant EP/S016562/1.} 
\\ {\tt M.E.Dyer@leeds.ac.uk}
}
\date{}
\maketitle

\bibliographystyle{plainurl}

\begin{abstract}
We consider the complexity of counting weighted graph homomorphisms
defined by a symmetric matrix $A$.
Each symmetric matrix $A$
defines a graph homomorphism function $Z_A(\cdot)$,
also known as the partition function.
Dyer and Greenhill~\cite{Dyer-Greenhill-2000}
established a complexity dichotomy of $Z_A(\cdot)$ for
symmetric $\{0, 1\}$-matrices $A$, and they further proved that its
\#P-hardness part also holds for bounded degree graphs.
Bulatov and Grohe~\cite{Bulatov-Grohe-2005} extended 
the  Dyer-Greenhill dichotomy
to nonnegative symmetric matrices $A$.
However, their hardness proof requires graphs of arbitrarily large degree,
and whether the
bounded degree part of the Dyer-Greenhill dichotomy can be 
extended has been an open problem for 15 years.
We resolve this open problem 
and prove that for nonnegative symmetric $A$,
either $Z_A(G)$ is in polynomial time for all graphs $G$,
or it is \#P-hard for bounded degree (and simple) graphs $G$.
We further extend the complexity dichotomy to include
nonnegative vertex weights.
Additionally, we prove that the \#P-hardness part
of the dichotomy by Goldberg et al.~\cite{Goldberg-et-al-2010} for $Z_A(\cdot)$
also holds for simple graphs,
where $A$ is any real symmetric matrix.
\end{abstract}

\thispagestyle{empty}
\clearpage
\setcounter{page}{1}


\section{Introduction}

The modern study of graph homomorphisms originates from the work by Lov\'asz and
others several decades
ago and has been a very active area~\cite{Lovasz-1967, GH-book}.
If  $G$ and $H$ are two graphs,
a  graph homomorphism (GH) is a mapping $f \colon V(G) \to V(H)$
that  preserves vertex adjacency, i.e.,
whenever $(u, v)$ is an edge in $G$, $(f(u), f(v))$ is also an edge in $H$.
%
Many combinatorial problems on graphs can be expressed as graph homomorphism
problems.
Well-known examples include the problems of finding a proper vertex coloring, vertex cover, independent set and clique.
For example, if $V(H) = \{0, 1\}$ with an edge between $0$ and $1$ and a loop at $0$,
then $f \colon V(G) \to \{0, 1\}$ is a graph homomorphism iff $f^{-1}(1)$ is an independent set in $G$;
 similarly, proper vertex colorings on $G$ using at
most $m$ colors correspond to homomorphisms
 from $G$ to $H = K_m$ (with no loops).

More generally, one can consider weighted graphs  $H$
and aggregate all homomorphisms
from $G$ to $H$ into a weighted sum.  This
is a powerful graph invariant which can express many
graph properties.
Formally, for a symmetric $m \times m$ matrix $A$, 
the \emph{graph homomorphism function} on a graph $G = (V, E)$
is defined as follows:
\[
Z_A(G) = \sum_{\xi: V \rightarrow [m]} \prod_{(u, v) \in E} A_{\xi(u), \xi(v)}.
\]
Note that if $H$ is unweighted, and  $A$ is its $\{0, 1\}$-adjacency 
matrix, then each product $\prod_{(u, v) \in E} A_{\xi(u), \xi(v)}$ is $0$ or $1$,
and is $1$ iff $\xi$ is a graph homomorphism.
Thus in this case  $Z_A(G)$ counts the number of homomorphisms
from $G$ to $H$.
One can further allow $H$ to have vertex weights.
In this case, we can similarly define the 
function $Z_{A, D}(\cdot)$ (see Definition~\ref{def:EVAL(A,D)}).


These sum-of-product functions  $Z_A(\cdot)$ and $Z_{A, D}(\cdot)$
are referred to as the \emph{partition functions} in statistical physics~\cite{baxter-6-8}.
Various special cases of GH have been studied there extensively,
which include the Ising, Potts, hardcore gas, Beach, Widom-Rowlinsom models, etc.~\cite{baxter-6-8}.

The computational complexity of $Z_A(\cdot)$
has been studied systematically.
Dyer and Greenhill~\cite{Dyer-Greenhill-2000, Dyer-Greenhill-corrig-2004} 
proved that, for a symmetric $\{0, 1\}$-matrix $A$,
$Z_A(\cdot)$ is either in polynomial time or \#P-complete,
and they gave a succinct condition for this complexity dichotomy:
if $A$ satisfies the condition then $Z_A(\cdot)$ is computable in
polynomial time (we also call it \emph{tractable}), otherwise
it is \#P-complete.
Bulatov and Grohe~\cite{Bulatov-Grohe-2005} 
(see also~\cite{Thurley-2009, Grohe-Thurley-2011}) 
generalized the Dyer-Greenhill dichotomy to
 $Z_A(\cdot)$ for nonnegative symmetric matrices $A$.
It was further extended by Goldberg et al.~\cite{Goldberg-et-al-2010} 
to arbitrary real symmetric matrices,
and finally by Cai, Chen and Lu~\cite{Cai-Chen-Lu-2013} to arbitrary 
complex symmetric matrices.
In the last two dichotomies,
the tractability criteria are not trivial to state.
Nevertheless, both tractability criteria are decidable in polynomial time
(in the size of $A$).

The definition of the partition function $Z_A(\cdot)$
can be easily extended to
directed graphs $G$ and arbitrary (not necessarily symmetric) matrices $A$
corresponding to directed edge weighted graphs $H$.
Concerning the complexity of counting directed GH,
we currently have the \emph{decidable} dichotomies
by Dyer, Goldberg and Paterson~\cite{Dyer-Goldberg-Paterson-2007}
for $\{0, 1\}$-matrices corresponding to (unweighted) simple acyclic graphs $H$,
and by Cai and Chen~\cite{Cai-Chen-2019} for all nonnegative 
matrices $A$.

Dyer and Greenhill in the same paper~\cite{Dyer-Greenhill-2000}
proved a stronger statement that
if a $\{0, 1\}$-matrix $A$ fails the tractability condition then
 $Z_A(G)$  is \#P-complete even when restricted to 
bounded degree graphs $G$.
We note that the complexity of GH for bounded degree graphs
is particularly interesting as much work has been done on the
approximate complexity of GH focused on bounded degree graphs
 and approximate algorithms
are achieved for them~\cite{DyerFJ02,Weitz06,Sly10,SinclairST12,LiLY13,Barvinok-book,Barvinok-Soberon-2017,Peters-Regts-2018,HelmuthPR19,Peters-Regts-2018}.
However, for fifteen years  
the worst case complexity for bounded degree graphs
in the Bulatov-Grohe dichotomy
was open.
Since this dichotomy is used essentially
in almost all subsequent work, e.g.,~\cite{Goldberg-et-al-2010,Cai-Chen-Lu-2013},
this has been a stumbling block.
 
Our main contribution in this paper is to resolve this 15-year-old open problem. 
We prove
that the \#P-hardness part of
 the Bulatov-Grohe dichotomy still holds
 for \emph{bounded degree graphs}.
It can be further strengthened to apply to 
bounded degree \emph{simple} graphs.
We actually prove a broader
dichotomy for $Z_{A, D}(\cdot)$,
where in addition to the nonnegative symmetric edge weight matrix $A$
there is also a nonnegative diagonal vertex weight matrix $D$.
We will give an explicit tractability condition 
such that, if $(A, D)$ satisfies the condition then
 $Z_{A, D}(G)$ is computable in polynomial time for all $G$, 
and if it fails the condition then
$Z_{A, D}(G)$ is  \#P-hard even restricted to
\emph{bounded degree simple graphs} $G$.
$Z_A(G)$ is the special case of  $Z_{A, D}(G)$ when $D$ is
the identity matrix.
Additionally, we prove that
the \#P-hardness part of the dichotomy by Goldberg et al.~\cite{Goldberg-et-al-2010}
for all real symmetric edge weight matrices $A$ still holds for \emph{simple graphs}.
(Although in this case, whether under the same condition
on $A$ the \#P-hardness still holds
 for bounded degree  graphs is  not resolved in the present paper.)

In order to prove the dichotomy theorem on bounded degree graphs,
we have to introduce a nontrivial extension of the well-developed
interpolation method~\cite{Valiant}.
%
We use some of the well-established techniques in this area of research
such as stretchings and thickenings. But the main innovation
is an overall design of the interpolation for
a more abstract target polynomial than $Z_{A, D}$. 
To carry out the proof there is an initial condensation step
where we combine vertices that have 
proportionately the same neighboring edge wrights
(technically defined by pairwise linear dependence) into a super vertex with
a combined vertex weight. Note that this creates vertex weights
even when initially all vertex weights are 1. When vertex weights are present, 
an approach in interpolation proof
 is to arrange things well so that
in the end one can redistribute  vertex weights to edge weights.
However, when edge weights are not 0-1,
any gadget design must deal with a quantity at 
each vertex that cannot be \emph{redistributed}.
This dependence has the form
$\sum_{j = 1}^{m_{\zeta(w)}} \alpha_{\zeta(w) j} \mu_{\zeta(w) j}^{\deg(w)}$,
resulting from combining pairwise linearly dependent rows and columns,
that depends on vertex degree $\deg(w)$
in a complicated way.
(We note that in the  0-1 case all $\mu_{\zeta(w) j} \in \{0, 1\}$, making
it in fact degree \emph{independent}.)

We overcome this difficulty by essentially introducing
a virtual level of interpolation---an interpolation to realize
some ``virtual gadget'' that cannot be physically realized, and yet
its ``virtual'' vertex weights are suitable for redistribution.
Technically we have to define an auxiliary graph $G'$,
and express
 the  partition function
in an extended framework, called
$Z_{\mathscr A, \mathscr D}$ on $G'$ (see Definition~\ref{def:EVAL(scrA,scrD)}).
In a typical interpolation proof, there is a polynomial
with coefficients that have a clear combinatorial meaning
defined in terms of $G$, usually consisting of certain
sums of exponentially many terms in some target partition function.
Here, we will define a target polynomial
with certain coefficients;
however these coefficients do not have a direct combinatorial meaning
in terms of $Z_{A, D}(G)$, but rather they only have
a  direct combinatorial meaning in terms of 
$Z_{\mathscr A, \mathscr D}$ on $G'$.
In 
a suitable  ``limiting'' sense, 
a certain aggregate of these coefficients
forms some useful quantity in the final result.
This introduces a concomitant
``virtual'' vertex weight which depends on the vertex degree that is ``just-right'' 
so that it can be redistributed to become part of the incident edge weight,
thus effectively killing the vertex weight.
This leads to a reduction from $Z_{C}(\cdot)$ 
(without vertex weight) to $Z_{A, D}(\cdot)$, for some $C$ that inherits
the hardness condition of $A$,
thus proving the  \#P-hardness  of the latter.
This high level description will be made clearer
in Section~\ref{sec:Hardness-proof}. 
The nature of the degree dependent vertex weight
introduces a substantial difficulty; 
in particular
a direct adaptation of the proof
in~\cite{Dyer-Greenhill-2000}
does not work.


%

Our extended vertex-weighted version of the 
Bulatov-Grohe dichotomy can be used
to correct a crucial gap in the proof by Thurley~\cite{Thurley-2010}
for a dichotomy for $Z_A(\cdot)$ with Hermitian edge 
weight matrices $A$, where this degree dependence was also at 
the root of the 
difficulty.~\footnote{In~\cite{Thurley-2010}, the proof of Lemma~4.22 
uses Lemma~4.24. In Lemma~4.24, $A$ is assumed to have
pairwise linearly independent rows while Lemma~4.22 does not assume this,
and the author appeals to
a twin reduction step in~\cite{Dyer-Greenhill-2000}. However, unlike
in the 0-1 case~\cite{Dyer-Greenhill-2000},  
such a step incurs degree dependent vertex weights.
This gap is fixed by  our Theorem~\ref{thm:bd-hom-nonneg}.}

%
%
%
%
%

\section{Preliminaries}

In order to state all our complexity results in the strict notion of Turing
computability, we adopt the standard model~\cite{Lenstra-1992} of computation for
partition functions, and require that all numbers be from an arbitrary but fixed
algebraic extension of $\mathbb{Q}$. We use
 $\bbR$ and $\bbC$ to denote the sets
of real and complex algebraic numbers.
Many statements remain true in other fields or rings
if arithmetic operations can be carried out efficiently
in a model of computation
(see~\cite{cai-chen-book} for more discussions on this issue).

%

For a positive integer $n$, we use $[n]$
to denote the set $\{1, \ldots, n \}$.
When $n = 0$, $[0] = \emptyset$.
We use $[m:n]$, where $m \le n$,
to denote $\{ m, m + 1, \ldots, n \}$.

In this paper, we consider undirected graphs
unless stated otherwise.
%
%
Following standard definitions,
the graph $G$ is allowed to have multiple edges but no loops.
(However, we will touch on this issue a few times when
$G$ is allowed to have loops.)
The graph $H$ can have multiple edges and loops,
or more generally, edge weights.
For the graph $H$, we treat its loops as edges.

An edge-weighted graph $H$ on $m$ vertices can be identified with
a symmetric $m \times m$ matrix $A$ in the obvious way.
We write this correspondence by $H = H_A$ and $A = A_H$.

\begin{definition}\label{def:EVAL(A)}
Let $A \in \bbC^{m \times m}$ be a symmetric matrix.
The problem $\EVAL(A)$ is defined as follows:
Given an undirected graph $G = (V, E)$, compute
\[
Z_A(G) = \sum_{\xi: V \rightarrow [m]} \prod_{(u, v) \in E} A_{\xi(u), \xi(v)}.
\]
\end{definition}
The function $Z_A(\cdot)$ is called a \emph{graph homomorphism function} or a \emph{partition function}.
When $A$ is a symmetric $\{0, 1\}$-matrix, i.e.,
when the graph $H = H_A$ is unweighted,
 $Z_A(G)$ counts
the number of homomorphisms from $G$ to $H$.
In this case, we denote $\EVAL(H) = \EVAL(A_H)$, and this problem is also
known as the \#$H$-coloring problem.


\begin{theorem}[Dyer and Greenhill \cite{Dyer-Greenhill-2000}]\label{thm:Dyer-Greenhill}
Let $H$ be a fixed undirected graph.
Then $\EVAL(H)$ is in polynomial time
if every connected component of $H$ is either
(1) an isolated vertex, or (2) a complete graph with all loops present,
or (3) a complete bipartite graph with no loops present.
Otherwise, the problem $\EVAL(H)$ is \#P-complete.
\end{theorem}

Bulatov and Grohe~\cite{Bulatov-Grohe-2005} extended Theorem~\ref{thm:Dyer-Greenhill} to
$\EVAL(A)$ where $A$ is a symmetric matrix with nonnegative entries.
In order to state their result, we need to define a few notions first.

We say
a nonnegative 
symmetric $m \times m$ matrix $A$ is rectangular if
there are pairwise disjoint nonempty subsets of $[m]$:
$T_1, \ldots, T_r, P_1, \ldots, P_s, Q_1, \ldots, Q_s$,
for some $r, s\ge 0$, such that $A_{i, j} > 0$ iff
\[
(i, j) \in \bigcup_{k \in [r]} (T_k \times T_k) \cup \bigcup_{l \in [s]} [(P_l \times Q_l) \cup (Q_l \times P_l)].
\]
We refer to $T_k \times T_k, P_l \times Q_l$ and $Q_l \times P_l$ as blocks of $A$.
Further, we say a nonnegative symmetric matrix $A$ is \textit{block-rank-$1$}
if $A$ is rectangular and every block of $A$ has rank one.

\begin{theorem}[Bulatov and Grohe \cite{Bulatov-Grohe-2005}]\label{thm:Bulatov-Grohe}
Let $A$ be a symmetric matrix with nonnegative entries.
Then $\EVAL(A)$ is in polynomial time if $A$ is block-rank-$1$,
and is \#P-hard otherwise.
\end{theorem}

There is a natural extension of $\EVAL(A)$
involving the use of vertex weights. Both
papers~\cite{Dyer-Greenhill-2000,Bulatov-Grohe-2005} use them
in their proofs.
A graph $H$ on $m$ vertices with vertex and edge weights is
 identified with
a symmetric $m \times m$ edge weight matrix $A$ and
a diagonal $m \times m$ vertex weight matrix $D =
\diag(D_1, \ldots, D_m)$ in a natural way.
Then the problem $\EVAL(A)$ can be generalized
to $\EVAL(A, D)$ for vertex-edge-weighted graphs.

\begin{definition}\label{def:EVAL(A,D)}
Let $A \in \bbC^{m \times m}$ be a symmetric matrix
and $D \in \bbC^{m \times m}$ a diagonal matrix.
The problem $\EVAL(A, D)$ is defined as follows:
Given an undirected graph $G = (V, E)$, compute
\[
Z_{A, D}(G) = \sum_{\xi: V \rightarrow [m]} \prod_{w \in V} D_{\xi(w)} \prod_{(u, v) \in E} A_{\xi(u), \xi(v)}.
\]
\end{definition}
Note that $\EVAL(A)$ is the special case $\EVAL(A, I_m)$.
We also need to define another $\EVAL$ problem where
the vertex weights are specified by the 
degree.

\begin{definition}\label{def:EVAL(A,frakD)}
Let $A \in \bbC^{m \times m}$ be a symmetric matrix and
$\mathfrak D = \{ D^{\llbracket i \rrbracket} \}_{i = 0}^\infty$ a sequence of diagonal matrices in $\bbC^{m \times m}$.
The problem $\EVAL(A, \mathfrak D)$ is defined as follows:
Given an undirected graph $G = (V, E)$, compute
\[
Z_{A, \mathfrak D}(G) = \sum_{\xi: V \rightarrow [m]} \prod_{w \in V} D_{\xi(w)}^{\llbracket \deg(w) \rrbracket} \prod_{(u, v) \in E} A_{\xi(u), \xi(v)}.
\]
\end{definition}

Finally, we need to define a general $\EVAL$ problem,
where the vertices and edges can individually take specific weights.
Let $\mathscr A$ be a set of (edge weight) $m \times m$ matrices and
$\mathscr D$ a set of diagonal (vertex weight) $m \times m$ matrices.
A GH-grid $\Omega = (G, \rho)$ consists of
a graph $G = (V, E)$ with possibly both directed and undirected edges,
and loops,
and $\rho$ assigns to each edge $e \in E$ or loop an 
$A^{(e)} \in \mathscr A$
and  to each vertex $v \in V$
a $D^{(v)} \in \mathscr D$.
(A loop is just an edge of the form $(v,v)$.)
 If $e \in E$
is a directed edge then the tail and head correspond to rows and columns of 
$A^{(e)}$, respectively;
if  $e \in E$
is an undirected edge then $A^{(e)}$ must be
symmetric. 

\begin{definition}\label{def:EVAL(scrA,scrD)}
The problem $\EVAL(\mathscr A, \mathscr D)$ is defined as follows:
Given a GH-grid $\Omega = \Omega(G)$, compute
\[
Z_{\mathscr A, \mathscr D}(\Omega) =
\sum_{\xi \colon V \to [m]} \prod_{w \in V} D_{\xi(w)}^{(w)} \prod_{e = (u, v) \in E} A_{\xi(u), \xi(v)}^{(e)}
\]
\end{definition}


We remark that $Z_{\mathscr A, \mathscr D}$ is introduced 
only as a tool to express a certain quantity in a ``virtual''
interpolation;
the dichotomy theorems do not
apply to this.
Defintions~\ref{def:EVAL(A,frakD)} and~\ref{def:EVAL(scrA,scrD)}
are carefully crafted in order to carry out
the \#P-hardness part of the proof of Theorem~\ref{thm:bd-hom-nonneg}.
Notice that the problem $\EVAL(\mathscr A, \mathscr D)$ generalizes both
problems $\EVAL(A)$ and $\EVAL(A, D)$, by taking $\mathscr A$ to
be a single symmetric matrix, and 
by taking $\mathscr D$ to be a single diagonal matrix.
But $\EVAL(A, \mathfrak D)$ is not naturally
expressible as  $\EVAL(\mathscr A, \mathscr D)$ because the
latter does not force the 
vertex-weight matrix on a vertex according to its degree.

We refer to $[m]$ as the domain of the corresponding $\EVAL$ problem.
If $\mathscr A = \{ A \}$ or $\mathscr D = \{ D \}$, then 
we simply write $Z_{A, \mathscr D}(\cdot)$ or $Z_{\mathscr A, D}(\cdot)$, respectively.
%
%

We use a superscript $(\Delta)$ and/or a  subscript $\simp$
to  denote the restriction of a corresponding
$\EVAL$ problem
to  degree-$\Delta$ bounded graphs
and/or simple graphs. 
E.g., $\EVAL^{(\Delta)}(A)$ denotes the problem $\EVAL(A)$
restricted to degree-$\Delta$ bounded graphs,
 $\EVAL_{\simp}(A, \mathfrak D)$ denotes the problem
$\EVAL(A, \mathfrak D)$ restricted to simple graphs,
and both restrictions apply in $\EVAL_{\simp}^{(\Delta)}(A, \mathfrak D)$.

Working within the framework of  $\EVAL(A, D)$,
we define an edge gadget to be a graph with two distinguished vertices, 
called $u^*$ and $v^*$.
An edge gadget $G = (V, E)$  has a signature (edge weight matrix) 
expressed by an $m \times m$ matrix $F$, where
\[
F_{i j} = \sum_{\substack{\xi \colon V \to [m] \\ \xi(u^*) = i,\, \xi(v^*) = j}} \prod_{z \in V \setminus \{ u^*, v^* \}} D_{\xi(z)} \prod_{(x, y) \in E} A_{\xi(x), \xi(y)}
\]
for $1 \le i, j \le m$.
When this gadget is placed in a graph identifying  $u^*$ and $v^*$
with two vertices  $u$ and $v$ in that graph, then $F$ is the signature
matrix for the pair $(u, v)$.
Note that the vertex weights corresponding to $u$ and $v$
are excluded from the product in the definition of $F$.
Similar definitions can be introduced for 
$\EVAL(A)$, $\EVAL(A, \mathfrak D)$ and $\EVAL(\mathscr A, \mathscr D)$.
%
%

We use $\le_{\mathrm T}^{\mathrm P}$ (and $\equiv_{\mathrm T}^{\mathrm P}$)
to denote polynomial-time Turing reductions (and equivalences, respectively).


Two simple 
operations are known as \emph{thickening} and \emph{stretching}.
%
Let $p, r \ge 1$ be integers.
A $p$-\emph{thickening} of an edge 
replaces it by $p$ parallel edges,
and
a $r$-\emph{stretching} replaces it by a path
of length $r$.
In both cases we retain the endpoints $u, v$.
The $p$-\emph{thickening} or $r$-\emph{stretching}
 of $G$ with respect to $F \subseteq E(G)$,
denoted respectively by $T_p^{(F)}(G)$   and $S_r^{(F)}(G)$,
are
obtained by $p$-\emph{thickening} or  $r$-\emph{stretching} 
 each edge from $F$, respectively.
Other edges, if any, are unchanged in both cases.
When $F = E(G)$, we call them the $p$-\emph{thickening} and $r$-\emph{stretching} of $G$
and denote them by $T_p(G)$ and $S_r(G)$, respectively.
$T_p e$ and $S_r e$ are the special cases when the graph consists of
a single edge $e$.
See Figure \ref{fig:thickening-stretching} for an illustration.
Thickenings and stretchings  can be combined in any order.
Examples are shown in Figure~\ref{fig:thickenings-stretchings-composition}.


\begin{figure}[t]
\begin{subfigure}{0.5\textwidth}
\centering
\includegraphics{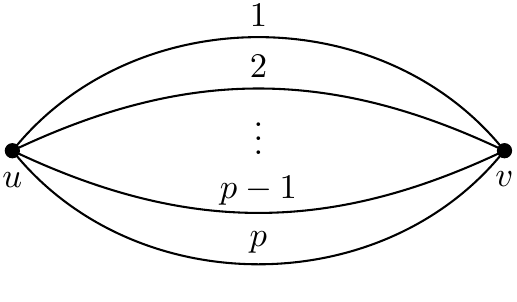}
\end{subfigure}\hfill
\begin{subfigure}{0.5\textwidth}
\centering
\includegraphics{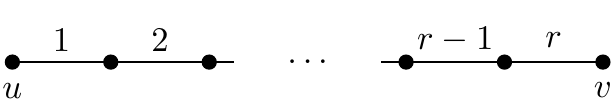}
\end{subfigure}
\caption{\label{fig:thickening-stretching}The thickening $T_p e$ and the stretching $S_r e$ of an edge $e = (u, v)$.}
\end{figure}


\begin{figure}[t]
\begin{subfigure}{0.6\textwidth}
\centering
\includegraphics{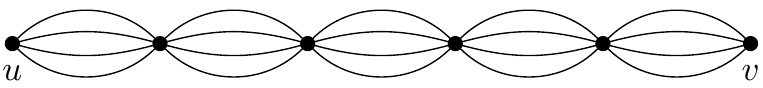}
\end{subfigure}\hfill
\begin{subfigure}{0.4\textwidth}
\centering
\includegraphics{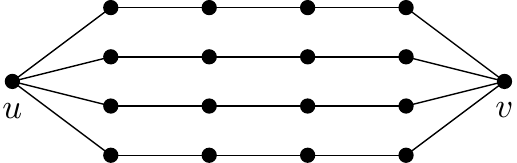}
\end{subfigure}
\caption{\label{fig:thickenings-stretchings-composition}The graphs $T_4 S_5 e$ (on the left) and $S_5 T_4 e$ (on the right) where $e = (u, v)$.}
\end{figure}

For a matrix $A$, we denote by $A^{\odot p}$ the matrix
obtained by replacing each entry of $A$ with its $p$th power.
Clearly, $Z_A(T_p G) = Z_{A^{\odot p}}(G)$ and $Z_A(S_r G) = Z_{A^r}(G)$.
More generally, for the vertex-weighted case,
we have $Z_{A, D}(T_p G) = Z_{A^{\odot p}, D}(G)$ and $Z_{A, D}(S_r G) = Z_{A (D A)^{r - 1}, D}(G)$.
Here $(D A)^0 = I_m$ if $A$ and $D$ are $m \times m$.

\section{Dichotomy for bounded degree graphs}

In addition to the Dyer-Greenhill dichotomy (Theorem~\ref{thm:Dyer-Greenhill}),
in the same paper~\cite{Dyer-Greenhill-2000} 
they also
proved that  the \#P-hardness part of their dichotomy 
holds for bounded degree graphs.
The bounded degree case of the Bulatov-Grohe dichotomy 
(Theorem~\ref{thm:Bulatov-Grohe}) 
was left open,
and all known proofs~\cite{Bulatov-Grohe-2005,Thurley-2009,Grohe-Thurley-2011} of its \#P-hardness
part require unbounded degree graphs.
All subsequent dichotomies that use the Bulatov-Grohe dichotomy,
e.g.,~\cite{Goldberg-et-al-2010,Cai-Chen-Lu-2013} also explicitly
or implicitly
(because of their dependence on the Bulatov-Grohe dichotomy) 
 require unbounded degree graphs.
In this paper, we extend the \#P-hardness part of the Bulatov-Grohe dichotomy 
to bounded degree graphs.
\begin{theorem}\label{thm:bd-hom-nonneg-weak}
Let $A$ be a symmetric nonnegative matrix.
If $A$ is not block-rank-$1$, then for some $\Delta > 0$,
the problem $\EVAL^{(\Delta)}(A)$ is \#P-hard.
\end{theorem}

The degree bound $\Delta$  proved
 in Theorem~\ref{thm:bd-hom-nonneg-weak} depends on $A$,
as is the case in Theorem~\ref{thm:Dyer-Greenhill}.
 The authors 
of~\cite{Dyer-Greenhill-2000} conjectured that a universal bound $\Delta =3$
works for Theorem~\ref{thm:Dyer-Greenhill}; whether
a universal bound exists for both
Theorems~\ref{thm:Dyer-Greenhill} and~\ref{thm:bd-hom-nonneg-weak} is open.
For general symmetric real or complex $A$, it is open
whether bounded degree versions of the dichotomies 
in~\cite{Goldberg-et-al-2010} and~\cite{Cai-Chen-Lu-2013} hold.
Xia~\cite{MingjiXia} proved that a  universal bound 
does not exist for complex symmetric matrices $A$, assuming \#P
does not collapse to P.

We prove a broader dichotomy than Theorem~\ref{thm:bd-hom-nonneg-weak},
which also includes arbitrary
nonnegative vertex weights.

\begin{theorem}\label{thm:bd-hom-nonneg}
Let $A$ and $D$ be $m \times m$ nonnegative matrices, where $A$ is symmetric, and $D$ is diagonal. 
Let $A'$ be the matrix obtained from $A$ by
 striking out rows and columns that correspond to
$0$ entries of $D$ on the diagonal.
If $A'$ is block-rank-$1$, then the problem $\EVAL(A, D)$ is in polynomial time.
Otherwise, for some $\Delta > 0$, the problem $\EVAL_{\simp}^{(\Delta)}(A, D)$ is \#P-hard.
\end{theorem}

Every $0$ entry of $D$ on the diagonal effectively nullifies
the corresponding domain element in $[m]$, so the problem
becomes an equivalent problem on the reduced domain. Thus, for a nonnegative diagonal
$D$, without
loss of generality, we may assume the domain has already been reduced
so that $D$ is positive diagonal. In what follows,  we will make this
assumption.

In Appendix~\ref{sec:Tractability-part},  we will prove the 
tractability part of
Theorem~\ref{thm:bd-hom-nonneg}. This follows easily from known results.
In Appendix~\ref{sec:Two-technical-lemmas}, 
we will present two technical lemmas, 
Lemma~\ref{lem:ADA-pairwise-lin-ind} and Lemma~\ref{lem:ADA-nondeg-thick}
 to be used in
Section~\ref{sec:Hardness-proof}. Finally, in Appendix~\ref{sec:Goldberg-et-al-2010-dichotomy} we prove 
Theorem~\ref{thm:EVAL-simp-interp},
showing that the \#P-hardness part of  the dichotomy for counting GH
by Goldberg et al.~\cite{Goldberg-et-al-2010} for
real symmetric matrix (with mixed signs) is also valid for  simple graphs.

\section{Hardness proof}\label{sec:Hardness-proof}

We proceed to prove the \#P-hardness part of Theorem~\ref{thm:bd-hom-nonneg}.
Let $A$ and $D$ be $m \times m$ matrices,
where $A$ is nonnegative symmetric but not block-rank-$1$, and $D$ is positive diagonal.
%
The first step is to eliminate pairwise linearly dependent rows and columns
of $A$.
(We will see that this step will naturally create nontrivial vertex weights
even if we initially start with the vertex unweighted case $D=I_m$.)

%

If $A$ has a zero row or column $i$, then
for any connected input graph $G$ other than a single isolated vertex,
no map $\xi: V(G) \rightarrow [m]$ having
 a nonzero contribution to $Z_{A, D}(G)$
can map any vertex of $G$ to $i$.  
So, by crossing out all zero rows and columns (they have the same index
set since $A$ is symmetric) 
we may assume that $A$ has no zero rows or columns.
We then delete the same set of rows and columns from $D$,
thereby expressing the problem
$\EVAL_{\simp}^{(\Delta)}(A, D)$ for $\Delta \ge 0$ on a smaller domain.
Also permuting the rows and columns of both $A$ and $D$ simultaneously by the same permutation
does not change the value of $Z_{A, D}(\cdot)$,
and so it does not change the complexity of
$\EVAL_{\simp}^{(\Delta)}(A, D)$ for $\Delta \ge 0$ either.
Having no zero rows and columns implies that pairwise linear dependence
is an equivalence relation, and so
we may assume that the pairwise linearly dependent rows and
columns of $A$ are contiguously arranged.
Then, after renaming the indices, the entries of $A$ are of the following form:
$A_{(i, j), (i', j')} = \mu_{i j} \mu_{i' j'} A'_{i, i'}$,
where $A'$ is a nonnegative symmetric $s \times s$ matrix
with all columns nonzero and pairwise linearly independent,
 $1 \le i, i' \le s$, $1 \le j \le m_i$, $1 \le j' \le m_{i'}$,
 $\sum_{i = 1}^s m_i = m$,
and all $\mu_{i j} > 0$.
We also rename the indices of the matrix $D$ so that
the diagonal entries of $D$
are of the following form:
$D_{(i, j), (i, j)} = \alpha_{i j} > 0$
for $1 \le i \le s$ and $1 \le j \le m_i$.
As $m \ge 1$ we get $s \ge 1$.

Then the partition function $Z_{A, D}(\cdot)$ can be written in 
a compressed form
\[
Z_{A, D}(G)  =
\sum_{\zeta: V(G) \rightarrow [s]} \left( \prod_{w \in V(G)} \sum_{j = 1}^{m_{\zeta(w)}} \alpha_{\zeta(w) j} \mu_{\zeta(w) j}^{\deg(w)} \right) \prod_{(u, v) \in E(G)} A'_{\zeta(u), \zeta(v)} = Z_{A', \mathfrak D}(G)
\]
where $\mathfrak D = \{ D^{\llbracket k \rrbracket}\}_{k = 0}^\infty$
with $D^{\llbracket k \rrbracket}_i = \sum_{j = 1}^{m_i} \alpha_{i j} \mu_{i j}^k > 0$ for $k \ge 0$ and $1 \le i \le s$.
Then all matrices in $\mathfrak D$ are positive diagonal.
Note the dependence on the vertex degree $\deg(w)$ for $w \in V(G)$.
Since the underlying graph $G$ remains unchanged,
this way we obtain the equivalence
$\EVAL_{\simp}^{(\Delta)}(A, D) \equiv_{\mathrm T}^{\mathrm P} \EVAL_{\simp}^{(\Delta)}(A', \mathfrak D)$ for any $\Delta \ge 0$.
Here the subscript $\simp$ can be included or excluded,
and the same is true for the superscript $(\Delta)$,
the statement remains true in all cases.
We also point out that
 the entries of the matrices $D^{\llbracket k \rrbracket} \in \mathfrak D$
are computable in polynomial time
in the input size of $(A, D)$ as well as in $k$.

\subsection{Gadgets $\mathcal P_{n, p}$ and $\mathcal R_{d, n, p}$}\label{subsec:Gadget-Rdnp}

We first introduce the  \emph{edge gadget}
 $\mathcal P_{n, p}$, for all $p, n \ge 1$.
It is  obtained by replacing each edge
of a  path of length $n$ by  the gadget in
Figure~\ref{fig:ADA-to-p-gadget-advanced}
from Lemma~\ref{lem:ADA-nondeg-thick}.
More succinctly  $\mathcal P_{n, p}$ is  $S_2 T_p S_n e$, where $e$ is an edge.

To define the gadget $\mathcal R_{d, n, p}$, for all $d, p, n \ge 1$, we start
with a cycle on $d$ vertices $F_1, \ldots, F_d$ (call it a $d$-cycle),
replace every edge of the  $d$-cycle by a copy of  $\mathcal P_{n, p}$, 
and append a dangling
edge at each vertex $F_i$ of the $d$-cycle.
To be specific, a $2$-cycle has
 two vertices with $2$ parallel edges between them,
and a $1$-cycle
is a loop on one vertex. The gadget $\mathcal R_{d, n, p}$
always has $d$ dangling edges.
Note that all
 $\mathcal R_{d, n, p}$ are loopless simple graphs (i.e., without
parallel edges or loops), for $d, n, p \ge 1$.
%
An example of a gadget $\mathcal R_{d, n, p}$
is shown in Figure~\ref{fig:d-gon-gagdet-simplified}.
For the special cases $d = 1, 2$,
examples of gadgets $\mathcal R_{d, n, p}$
can be seen in Figure~\ref{fig:d=1,2-gadgets}.

\begin{figure}[t]
\centering
\includegraphics{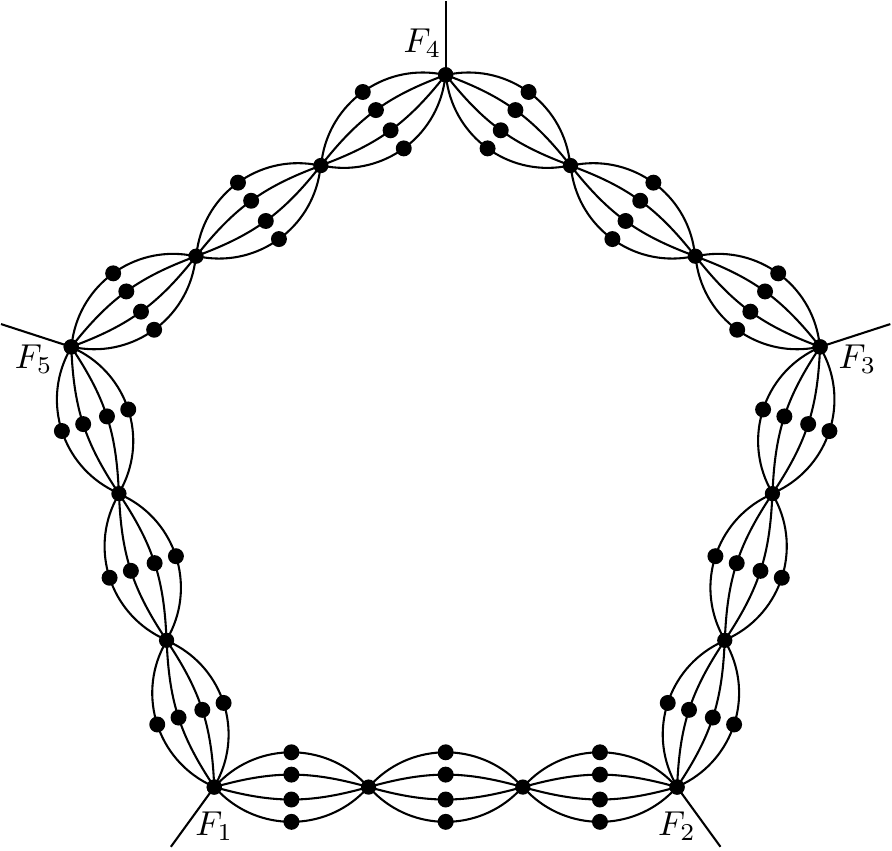}
\caption{\label{fig:d-gon-gagdet-simplified}The gadget $\mathcal R_{5, 3, 4}$.}
\end{figure}

\begin{figure}[t]
\setbox1=\hbox{\includegraphics[scale=0.73]{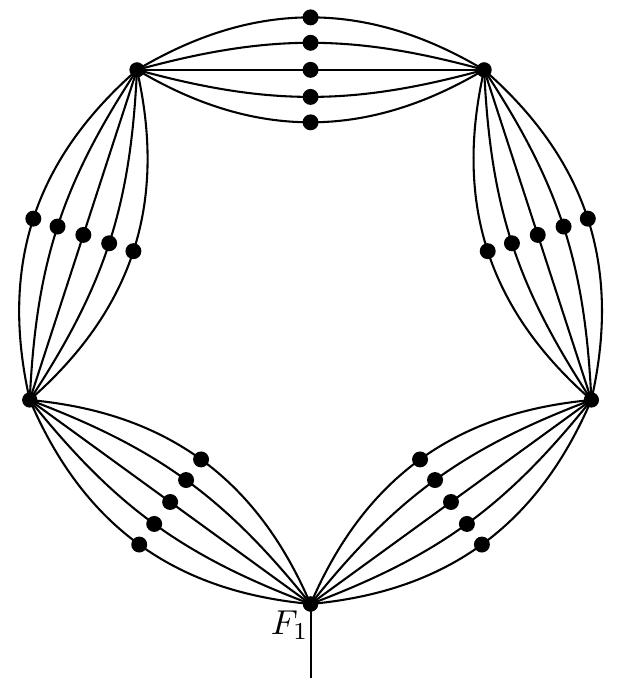}}
\setbox2=\hbox{\includegraphics[scale=0.73]{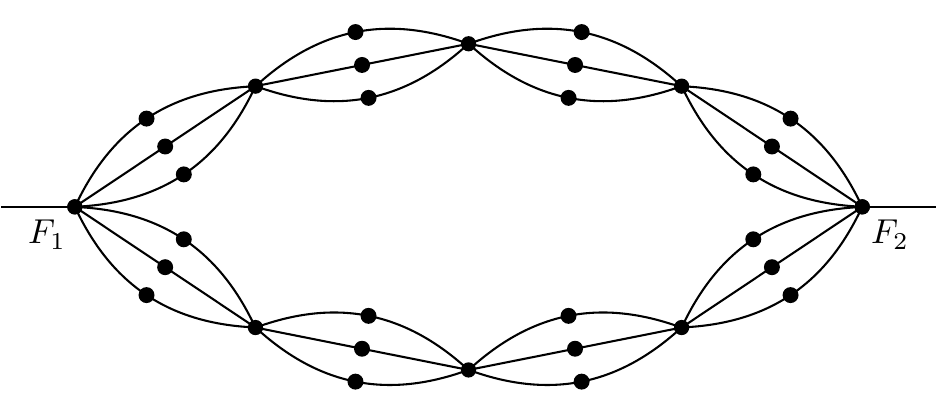}} 
\begin{subfigure}[t]{0.5\textwidth} 
\centering
\includegraphics[scale=0.73]{figure6.pdf} 
\subcaption{$\mathcal R_{1, 5, 5}$\label{subfig-1:d=1}}
\end{subfigure}\hfill
\begin{subfigure}[t]{0.5\textwidth} 
\centering
\raisebox{0.5\ht1}{\raisebox{-0.5\ht2}{
\includegraphics[scale=0.73]{figure7.pdf} 
}} 
\subcaption{$\mathcal R_{2, 4, 3}$\label{subfig-2:d=2}}
\end{subfigure}
\caption{\label{fig:d=1,2-gadgets}Examples of gadgets $\mathcal R_{d, n, p}$ for $d = 1, 2$}
\end{figure}

We note that
vertices in $\mathcal P_{n, p}$ have degrees at most $2p$,
and vertices in $\mathcal R_{d, n, p}$  have degrees at most $2p+1$,
taking into account the dangling edges. Clearly
$|V({\mathcal R_{d, n, p}})| = d n (p+1)$ and 
$|E({\mathcal R_{d, n, p}})| = (2np+1) d$, including the  dangling edges.

By Lemma~\ref{lem:ADA-nondeg-thick},
we can fix some $p \ge 1$ such that $B = (A' D^{\llbracket 2 \rrbracket} A')^{\odot p}$ is nondegenerate, where the superscript $\llbracket 2 \rrbracket$
is from the stretching operator $S_2$ which creates those degree
$2$ vertices, and the  superscript $\odot p$ is
 from the thickening operator $T_p$, followed by $S_2$, which creates those
parallel paths of length $2$.
%
The edge gadget $\mathcal P_{n, p}$ has the edge weight matrix
\begin{align}
L^{(n)} &=
\underbrace{B D^{\llbracket 2 p \rrbracket} B \ldots B D^{\llbracket 2 p \rrbracket} B}_{D^{\llbracket 2 p \rrbracket} \text{ appears } n - 1 ~\ge~ 0 \text{ times}} = B (D^{\llbracket 2 p \rrbracket} B)^{n - 1} \label{Pnp-edgeweightmatrix1} \\
&= (D^{\llbracket 2 p \rrbracket})^{-1 / 2} ((D^{\llbracket 2 p \rrbracket})^{1 / 2} B (D^{\llbracket 2 p \rrbracket})^{1 / 2})^n (D^{\llbracket 2 p \rrbracket})^{-1 / 2},
\label{Pnp-edgeweightmatrix2}
\end{align}
where in the notation $L^{(n)}$
we suppress the index $p$.
The  $n-1$ occurrences of
 $D^{\llbracket 2 p \rrbracket}$ in (\ref{Pnp-edgeweightmatrix1}) are due to
those $n-1$ vertices of degree $2p$.
Here $(D^{\llbracket 2 p \rrbracket})^{1 / 2}$ is a diagonal matrix
with the positive square roots
of the corresponding entries
of $D^{\llbracket 2 p \rrbracket}$ on the main diagonal,
and $(D^{\llbracket 2 p \rrbracket})^{-1 / 2}$ is its inverse.
The vertices $F_i$ are of degree $2 p + 1$ each,
but the contributions by its vertex weights are
not included in $L^{(n)}$.

The constraint function induced by  $\mathcal R_{d, n, p}$ 
is more complicated to write down. When it is placed as a part of
a graph, for any given assignment
to the $d$ vertices $F_i$, we can express the contribution of
the gadget  $\mathcal R_{d, n, p}$ in terms of $d$ copies of
$L^{(n)}$, \emph{together with} the vertex weights incurred at the
 $d$ vertices $F_i$ which will depend on their degrees.

\subsection{Interpolation using $\mathcal R_{d, n, p}$}\label{subsec:Interpolation}

Assume for now that $G$ does not contain isolated vertices.
We will  replace every vertex $u \in V(G)$ of degree  $d = d_u = \deg(u) \ge 1$ 
by a copy of  $\mathcal R_{d, n, p}$, for all $n, p \ge 1$.
The replacement operation 
can
be described in two steps: In step one, each $u \in V(G)$  is replaced
by a $d$-cycle on vertices $F_1, \ldots, F_d$,
each having a dangling edge attached.
The  $d$ dangling edges will be identified one-to-one with
the $d$ incident edges at $u$.
If $u$ and $v$ are adjacent vertices in $G$, then
the edge $(u, v)$ in $G$ will be replaced by merging a
pair of dangling edges, one from the $d_u$-cycle
and one from the $d_v$-cycle.
Thus in step one we obtain a graph $G'$, which basically replaces every
vertex $u \in V(G)$  by a cycle of $\deg(u)$ vertices.
Then in step two, for every cycle in $G'$  that
corresponds to some $u \in V(G)$ we replace each edge
on the cycle by a copy of the edge gadget
 $\mathcal P_{n, p}$.

Let $G_{n, p}$ denote the graph obtained from $G$ by the 
replacement procedure above.
Since all gadgets $\mathcal R_{d, n, p}$ are loopless simple graphs,
so are $G_{n, p}$ for all $n, p \ge 1$,
even if $G$ has multiple edges
(or had multiloops, 
if we view a loop as adding degree $2$ to the incident vertex).
As a technical remark, if $G$ contains vertices of degree $1$, 
then the intermediate graph $G'$ has loops but all graphs $G_{n, p}$
($n, p \ge 1$) do not.
Also note that all vertices in $G_{n, p}$ have degree
at most
$2 p + 1$,
which is independent of $n$.

Next, it is not hard to see that
\begin{gather*}
|V(G_{n, p})| = \sum_{u \in V(G)} d_u n (p + 1)  
= 2 n (p + 1) |E(G)|, \\
|E(G_{n, p})| = |E(G)| + \sum_{u \in V(G)} 2 n p d_u  
= (4 n p + 1) |E(G)|.
\end{gather*}
Hence  the size of 
the graphs $G_{n, p}$ is polynomially bounded in the size of $G$, $n$ and $p$.

Since we chose a fixed $p$,
and will choose $n$ to be bounded by a polynomial in the size of $G$,
whenever something is computable in polynomial time in $n$,
it is also computable in polynomial time in the size of $G$
(we will simply say in polynomial time).

We consider $Z_{A', \mathfrak D}(G)$, and substitute
$G$ by $G_{n, p}$.
We will make use of the edge weight matrix $L^{(n)}$ of  $\mathcal P_{n, p}$
in (\ref{Pnp-edgeweightmatrix2}).
%
The vertices $F_i$ are of degree $2 p + 1$  each in $G_{n, p}$,
so will each contribute
a vertex weight according to the diagonal 
matrix $D^{\llbracket 2 p + 1 \rrbracket}$
to the partition function, which are not included in $L^{(n)}$,
but now must be accounted for in $Z_{A', \mathfrak D}(G_{n, p})$.

Since $B$ is real symmetric and $D^{\llbracket 2 p \rrbracket}$ is positive diagonal,
the matrix 
\[\widetilde B = (D^{\llbracket 2 p \rrbracket})^{1 / 2} B (D^{\llbracket 2 p \rrbracket})^{1 / 2}\]
 is real symmetric.
Then $\widetilde B$ is orthogonally diagonalizable over $\mathbb{R}$, i.e.,
there exist a real orthogonal matrix $S$
and a real diagonal matrix $J = \diag(\lambda_i)_{i = 1}^s$ such that $\widetilde B = S^T J S$.
Then ${\widetilde B}^n = S^T J^n S$ so the edge weight matrix for 
 $\mathcal P_{n, p}$ becomes
\[
L^{(n)} 
= (D^{\llbracket 2 p \rrbracket})^{-1 / 2} {\widetilde B}^n (D^{\llbracket 2 p \rrbracket})^{-1 / 2}
= (D^{\llbracket 2 p \rrbracket})^{-1 / 2} S^T J^n S (D^{\llbracket 2 p \rrbracket})^{-1 / 2}.
\]
Note that $L^{(n)}$ as a matrix is defined for any $n \ge 0$,
and $L^{(0)} = (D^{\llbracket 2 p \rrbracket})^{-1}$,
even though there is no physical gadget $\mathcal P_{0, p}$ that corresponds to it.
However, it is precisely this ``virtual'' gadget we wish to
``realize'' by interpolation.

Clearly, $\widetilde B$ is nondegenerate as $B$ and $(D^{\llbracket 2 p \rrbracket})^{1/2}$ both are, and so is $J$.
Then all $\lambda_i \ne 0$.
We can also write $L_{i j}^{(n)} = \sum_{\ell = 1}^s a_{i j \ell} \lambda_\ell^n$
for every $n \ge 0$ and some real $a_{i j \ell}$'s
which depend on $S$, $D^{\llbracket 2 p \rrbracket}$, but not on
$J$ and $n$, for all $1 \le i, j, \ell \le s$.
By the formal expansion of the symmetric matrix $L^{(n)}$ above,
we have $a_{i j \ell} = a_{j i \ell}$.
Note that for all $n, p \ge 1$,
the gadget $\mathcal R_{d_v, n, p}$ for $v \in V(G)$ employs  exactly
$d_v$ copies of $\mathcal P_{n, p}$. Let
$t = \sum_{v \in V(G)} d_v = 2 |E|$; this  is precisely the number of
edge gadgets $\calP_{n, p}$ in $G_{n, p}$.

In the evaluation of the partition function $Z_{A', \mathfrak D}(G_{n, p})$,
we stratify the vertex assignments in $G_{n, p}$ as follows.
Denote by $\kappa = (k_{i j})_{1 \le i \le j \le s}$
a tuple of nonnegative integers, where the indexing is over all
$s(s+1)/2$ ordered pairs $(i, j)$. There are
a total of $\binom{t + s (s + 1) / 2 - 1}{s (s + 1) / 2 - 1}$
such tuples that satisfy $\sum_{1 \le i \le j \le s}  k_{i j} = t$.
For a fixed $s$, this is a polynomial in $t$, and thus
a polynomial in the size of $G$. 
Denote by $\mathcal K$ the set of all such 
tuples $\kappa$.
We will stratify all vertex assignments in $G_{n, p}$ 
by $\kappa\in  \mathcal K$, namely all assignments
such that there are exactly $k_{i j}$ many constituent edge gadgets $\mathcal P_{n, p}$ with the two end points (in either order of the end points)
 assigned $i$ and $j$ respectively.


For each $\kappa \in \mathcal K$, the edge gadgets $\mathcal P_{n, p}$
in total contribute $\prod_{1 \le i \le j \le s} (L_{i j}^{(n)})^{k_{i j}}$
to the partition function $Z_{A', \mathfrak D}(G_{n, p})$.
If we factor this product out for each $\kappa \in \mathcal K$, we can express
$Z_{A', \mathfrak D}(G_{n, p})$ as a linear
combination of these products over all $\kappa \in \mathcal K$,
with  polynomially many coefficient values  $c_\kappa$
that are independent of all edge gadgets $\mathcal P_{n, p}$.
Another way to define these coefficients $c_\kappa$ is to think in terms of
 $G'$: For any $\kappa = (k_{i j})_{1 \le i \le j \le s}
 \in \mathcal K$,
we say a vertex assignment on $G'$ is consistent with $\kappa$ 
if it assigns exactly  $k_{i j}$ many cycle edges  of  $G'$
(i.e., those that belong to the cycles that replaced vertices in $G$) 
as ordered
pairs of vertices to the values $(i, j)$ or $(j,i)$. 
(For any loop in $G'$, as a cycle of length $1$ that came from 
a degree $1$ vertex of $G$, it can only be assigned $(i,i)$ for some
 $1 \le i \le s$.)
Let  $L'$ be any symmetric edge signature to be
assigned on each of these
cycle edges in $G'$,
and keep  the edge signature $A'$ on the
merged dangling edges between any two such cycles, 
and the suitable vertex
weights specified by $\mathfrak D$, namely each vertex receives its vertex weight according to  $D^{\llbracket 2 p +1 \rrbracket}$.
Then $c_\kappa$
is the sum, over all assignments
consistent with $\kappa$, of the products of  all 
edge weights and vertex weights \emph{other than}
the contributions by $L'$, in the evaluation of
the partition function on $G'$.
In other words, for each $\kappa \in \mathcal K$,
\[
c_\kappa = \sum_{\substack{\zeta \colon V(G') \to [s] \\ \zeta \text{ is consistent with } \kappa}} \prod_{w \in V(G')} D_{\zeta(w)}^{\llbracket 2 p + 1 \rrbracket} \prod_{(u, v) \in \widetilde E} A'_{\zeta(u), \zeta(v)},
\]
where $\widetilde E \subseteq E(G')$ are the non-cycle edges of $G'$ that are in $1$-$1$ correspondence with $E(G)$.

In particular, the values $c_\kappa$ are independent of $n$.
Thus for some polynomially many values $c_\kappa$, where $\kappa
\in \mathcal K$, we have
\begin{equation}\label{stratification-isolating-L}
Z_{A', \mathfrak D}(G_{n, p}) = \sum_{\kappa \in \mathcal K} c_\kappa \prod_{1 \le i \le j \le s} (L_{i j}^{(n)})^{k_{i j}}
= \sum_{\kappa \in \mathcal K} c_\kappa \prod_{1 \le i \le j \le s} (\sum_{\ell = 1}^s a_{i j \ell} \lambda_\ell^n)^{k_{i j}}. 
\end{equation}
Expanding out the last sum and rearranging the terms, for some
values $b_{i_1, \ldots, i_s}$  independent of $n$, we get
\[
Z_{A', \mathfrak D}(G_{n, p})
= \sum_{\substack{i_1 + \ldots + i_s = t \\ i_1, \ldots, i_s \ge 0}} b_{i_1, \ldots, i_s} ( \prod_{j = 1}^s \lambda_j^{i_j} )^n
\]
for all $n \ge 1$.

This represents a linear system with the unknowns $b_{i_1, \ldots, i_s}$ 
with the rows indexed by $n$. 
The number of unknowns is clearly $\binom{t + s - 1}{s - 1}$
which is polynomial in the size of the input graph $G$ since $s$ is a constant.
The values $\prod_{j = 1}^s \lambda_j^{i_j}$
can be clearly computed in polynomial time.

We show how to compute the value 
\[\sum_{\substack{i_1 + \ldots + i_s = t \\ i_1, \ldots, i_s \ge 0}} b_{i_1, \ldots, i_s}\]
from the values $Z_{A', \mathfrak D}(G_{n, p}),\, n \ge 1$ in polynomial time.
The coefficient matrix of this system is a Vandermonde matrix.
However, it can have repeating columns so it might not be of full rank
because the coefficients $\prod_{j = 1}^s \lambda_j^{i_j}$
do not have to be pairwise distinct.
However, when they are equal, say,
$\prod_{j = 1}^s \lambda_j^{i_j} =
\prod_{j = 1}^s \lambda_j^{i'_j}$, we replace
the corresponding unknowns $b_{i_1, \ldots, i_s}$ and $b_{i'_1, \ldots, i'_s}$
with their sum as a new variable.
Since all $\lambda_i \ne 0$,
we have a Vandermonde system of full rank after all such combinations.
Therefore we can solve this linear system
in polynomial time and find the desired value $\displaystyle \sum_{\scriptsize \substack{i_1 + \ldots + i_s = t \\ i_1, \ldots, i_s \ge 0}} b_{i_1, \ldots, i_s}$.

Now we will consider a problem in the framework of $Z_{\mathscr A, \mathscr D}$
according to Definition~\ref{def:EVAL(scrA,scrD)}.
Let $G_{0, p}$ be the (undirected) GH-grid,
with the underlying graph $G'$,
and every edge of the cycle in $G'$ corresponding 
to a vertex in $V(G)$ is assigned
the edge weight matrix $(D^{\llbracket 2 p \rrbracket})^{-1}$,
and we keep the vertex-weight matrices $D^{\llbracket 2 p + 1 \rrbracket}$ at all vertices $F_i$.
The other edges, i.e., the original edges of $G$, each keep the assignment
of the edge weight matrix $A'$.
(So in the specification of  $Z_{\mathscr A, \mathscr D}$,
we have $\mathscr A = \{(D^{\llbracket 2 p \rrbracket})^{-1}, A'\}$,
and $\mathscr D = \{D^{\llbracket 2 p + 1 \rrbracket}\}$.
We note that $G'$ may have loops, and
Definition~\ref{def:EVAL(scrA,scrD)} specifically allows this.)
Then
\[
Z_{\{ (D^{\llbracket 2 p \rrbracket})^{-1}, A' \}, D^{\llbracket 2 p + 1 \rrbracket}}(G_{0, p})
= \sum_{\substack{i_1 + \ldots + i_s = t \\ i_1, \ldots, i_s \ge 0}} b_{i_1, \ldots, i_s} ( \prod_{j = 1}^s \lambda_j^{i_j} )^0
= \sum_{\substack{i_1 + \ldots + i_s = t \\ i_1, \ldots, i_s \ge 0}} b_{i_1, \ldots, i_s}
\]
and we have just computed this value in polynomial time
in the size of $G$ from the values $Z_{A', \mathfrak D}(G_{n, p})$, for $n \ge 1$.
In other words, we have achieved it by querying
the oracle $\EVAL(A', \mathfrak D)$ on the instances $G_{n, p}$, for $n \ge 1$,
in polynomial time.


Equivalently, we have shown that we can simulate
a virtual ``gadget'' $\mathcal R_{d, 0, p}$
replacing every occurrence of $\mathcal R_{d, n, p}$
in  $G_{n, p}$ in polynomial time.
The virtual gadget $\mathcal R_{d, 0, p}$ has
the edge signature $(D^{\llbracket 2 p \rrbracket})^{-1}$ in place of
$(D^{\llbracket 2 p \rrbracket})^{-1 / 2} {\widetilde B}^n
(D^{\llbracket 2 p \rrbracket})^{-1 / 2}$ in each $\mathcal P_{n, p}$, since
\[
(D^{\llbracket 2 p \rrbracket})^{-1 / 2} {\widetilde B}^0
(D^{\llbracket 2 p \rrbracket})^{-1 / 2}
= (D^{\llbracket 2 p \rrbracket})^{-1 / 2} I_s (D^{\llbracket 2 p \rrbracket})^{-1 / 2} = (D^{\llbracket 2 p \rrbracket})^{-1}.
\]
Additionally, each $F_i$ retains the vertex-weight contribution with the matrix $D^{\llbracket 2 p + 1 \rrbracket}$ in $\mathcal R_{d, 0, p}$.
We view it as having ``virtual'' degree $2 p + 1$.
This precisely results in the GH-grid $G_{0, p}$.

However, even though  $G_{0, p}$
still retains the cycles,
since $(D^{\llbracket 2 p \rrbracket})^{-1}$ is a diagonal matrix, each
vertex $F_i$
in a cycle is forced to receive the same vertex 
assignment value in the domain set $[s]$;
all other vertex assignments contribute
 zero in the evaluation of $Z_{\{ (D^{\llbracket 2 p \rrbracket})^{-1}, A' \}, D^{\llbracket 2 p + 1 \rrbracket}}(G_{0, p})$.
This can be easily seen by traversing the vertices  $F_1, \ldots, F_d$
in a cycle.
Hence we can view each cycle employing the
virtual gadget $\mathcal R_{d, 0, p}$ as a single vertex
that contributes only a diagonal matrix of positive vertex weights
$P^{\llbracket d \rrbracket} = (D^{\llbracket 2 p + 1 \rrbracket} (D^{\llbracket 2 p \rrbracket})^{-1})^d$, where $d$ is  the vertex degree in $G$.
Contracting all  the cycles
to a single vertex each,
we arrive at the original graph $G$.
Let $\mathfrak P = \{ P^{\llbracket i \rrbracket} \}_{i = 0}^\infty$,
where we let $P^{\llbracket 0 \rrbracket} = I_s$, and for $i>0$,
we have  $P^{\llbracket i \rrbracket}_j = w_j^i$
where $w_j = \sum_{k = 1}^{m_j} \alpha_{j k} \mu_{j k}^{2 p + 1} / \sum_{k = 1}^{m_j} \alpha_{j k} \mu_{j k}^{2 p} > 0$
for $1 \le j \le s$.
This shows that we now can interpolate
the value $Z_{A', \mathfrak P}(G)$
using the values $Z_{A', \mathfrak D}(G_{n, p})$
in polynomial time in the size of $G$.
The graph $G$ is arbitrary but without isolated vertices here.
We show next how to deal with the case when $G$ has isolated vertices.

Given an arbitrary graph $G$, assume it has $h \ge 0$ isolated vertices.
Let $G^*$ denote the graph obtained from $G$ by their removal.
Then $G^*$ is of size not larger than $G$ and $h \le |V(G)|$.
Obviously, $Z_{A', \mathfrak P}(G) = (\sum_{i = 1}^m P_i^{\llbracket 0 \rrbracket})^h Z_{A', \mathfrak P}(G^*) = s^h Z_{A', \mathfrak P}(G^*)$.
Here the integer $s$ is a constant, so the factor $s^h > 0$ can be easily computed in polynomial time.
Thus, knowing the value $Z_{A', \mathfrak P}(G^*)$
we can compute the value $Z_{A', \mathfrak P}(G)$ in polynomial time.
Further, since we only use the graphs $G_{n, p}, n \ge 1$ during the interpolation,
each being simple of degree at most $2 p + 1$,
combining it with the possible isolated vertex removal step,
we conclude $\EVAL(A', \mathfrak P) \le_{\mathrm T}^{\mathrm P} \EVAL_{\simp}^{(2 p + 1)}(A', \mathfrak D)$.

Next, it is easy to see that for an arbitrary graph $G$
\begingroup
\allowdisplaybreaks
\begin{align*}
Z_{A', \mathfrak P}(G) &= \sum_{\zeta: V(G) \rightarrow [s]} \prod_{z \in V(G)} P^{\llbracket \deg(z) \rrbracket}_{\zeta(z)} \prod_{(u, v) \in E(G)} A'_{\zeta(u), \zeta(v)} \\
&= \sum_{\zeta: V(G) \rightarrow [s]} \prod_{z \in V(G)} w_{\zeta(z)}^{\deg(z)} \prod_{(u, v) \in E(G)} A'_{\zeta(u), \zeta(v)} \\
&= \sum_{\zeta: V(G) \rightarrow [s]} \prod_{(u, v) \in E(G)} w_{\zeta(u)} w_{\zeta(v)} A'_{\zeta(u), \zeta(v)} \\
&= \sum_{\zeta: V(G) \rightarrow [s]} \prod_{(u, v) \in E(G)} C_{\zeta(u), \zeta(v)} = Z_C(G).
\end{align*}
\endgroup
Here $C$ is an $s \times s$ matrix with
the entries $C_{i j} = A'_{i j} w_i w_j$ where $1 \le i, j \le s$.
Clearly, $C$ is a nonnegative symmetric matrix.
In the above chain of equalities,
we were able to \emph{redistribute the weights $w_i$
and $w_j$ into the edge weights} $A'_{i j}$
which resulted in the edge weights $C_{i j}$,
so that precisely each edge $\{u,v\}$ in $G$ gets two factors
$w_{\zeta(u)}$ and $w_{\zeta(v)}$ since the vertex weights at  $u$ and $v$
were $w_{\zeta(u)}^{\deg(u)}$ and $w_{\zeta(v)}^{\deg(v)}$ respectively.
(This is a crucial step in our proof.)
Because the underlying graph $G$ is arbitrary, it follows that
$\EVAL(A', \mathfrak P) 
\equiv_{\mathrm T}^{\mathrm P} \EVAL(C)$.
Combining this with the previous $\EVAL$-reductions and equivalences,
we obtain
\[
\EVAL(C) \equiv_{\mathrm T}^{\mathrm P} \EVAL(A', \mathfrak P) \le_{\mathrm T}^{\mathrm P} \EVAL_{\simp}^{(2 p + 1)}(A', \mathfrak D) \equiv_{\mathrm T}^{\mathrm P} \EVAL_{\simp}^{(2 p + 1)}(A, D),
\]
so that $\EVAL(C) \le_{\mathrm T}^{\mathrm P} \EVAL_{\simp}^{(\Delta)}(A, D)$,
by taking $\Delta = 2 p + 1$.

Remembering that our goal is to prove
the \#P-hardness for the matrices $A, D$
not satisfying the tractability conditions of Theorem~\ref{thm:bd-hom-nonneg},
we finally use the assumption that $A$ is not block-rank-$1$.
Next, noticing that all $\mu_{i j} > 0$,
by construction $A'$ is not block-rank-$1$ either.
Finally, because all $w_i > 0$ nor is $C$ block-rank-$1$
implying that $\EVAL(C)$ is \#P-hard by Theorem~\ref{thm:Bulatov-Grohe}.
Hence $\EVAL_{\simp}^{(2 p + 1)}(A, D)$ is also \#P-hard.
This completes the proof of the 
\#P-hardness part of Theorem~\ref{thm:bd-hom-nonneg}.

We remark that
one important step in our interpolation proof happened 
at the stratification step 
before (\ref{stratification-isolating-L}).
In the proof we have the goal of redistributing
vertex weights to edge weights; but this redistribution is 
sensitive to the degree of the vertices.
This led us to define
the auxiliary graph $G'$ and the coefficients $c_\kappa$.
Usually in an interpolation proof there are some coefficients 
that have a clear combinatorial meaning in terms of the
original problem instance. 
Here these values $c_\kappa$ do not have a clear combinatorial meaning in terms of 
$Z_{A', \mathfrak D}(G)$, rather they are defined in terms of
an intermediate problem instance $G'$, which is neither $G$ nor the
actual constructed graphs $G_{n, p}$.
It is only in a ``limiting'' sense that a certain combination
of these values  $c_\kappa$ 
 allows us to compute $Z_{A', \mathfrak D}(G)$.

\appendix
\section{Tractability part}\label{sec:Tractability-part}

The tractability part of Theorem~\ref{thm:bd-hom-nonneg} follows easily from
known results. For completeness we outline a proof here.
Let $A$ and $D$ be $m \times m$ matrices,
where $A$ is nonnegative symmetric block-rank-$1$ and $D$ is positive diagonal.

First, $Z_{A, D}(G)$ can be reduced to the
connected components  $G_1, \ldots, G_t$  of $G$, 
\[
Z_{A, D}(G) = \prod_{i = 1}^t Z_{A, D}(G_i),
\]
so we may as well assume  $G$ is connected.
We permute the rows and columns of $A, D$ by the same permutation
and then cross out zero rows and columns of $A$. This 
does not change  $Z_{A, D}$. We may
assume that $A = \diag(A_i)_{i = 1}^k$ is block diagonal
with nonzero blocks $A_1, \ldots, A_k$, where each block $A_i$
is either a symmetric matrix of rank $1$ with no zero entries,
or a symmetric bipartite matrix of the form
$\left( \begin{smallmatrix} 0 & B \\ B^T & 0 \end{smallmatrix} \right)$
where $B$ has rank $1$ and no zero entries.
Then we can write $D = \diag(D_i)_{i = 1}^k$ 
where each $D_i$ is positive diagonal of the corresponding size.
As $A$ is block diagonal and $G$ is connected,
\[
Z_{A, D}(G) = \sum_{i = 1}^k Z_{A_i, D_i}(G).
\]
So we may as well assume that $A$ is one of these blocks.
Also let $D = \diag(\alpha_i)_{i = 1}^m$.


%
\begin{enumerate}
\item[1)] $A$ is a symmetric matrix of rank $1$ with no zero entries.
We can write $A = x^T x$ for some positive row vector $x = (x_i)_{i = 1}^m$.
Then
\begin{align*}
Z_{A, D}(G) = \prod_{u \in V(G)} \sum_{i = 1}^m \alpha_i x_i^{\deg(u)}.
\end{align*}

\item[2)]  $A = \left( \begin{smallmatrix} 0 & B \\ B^T & 0 \end{smallmatrix} \right)$,
where $B$ is $\ell \times (m - \ell)$
(for some $1 \le \ell < m$) has rank $1$ and no zero entries.
We can write $B = x^T y$ for some positive
row vectors $x = (x_i)_{i = 1}^\ell$ and $y = (y_j)_{j = \ell+1}^{m}$.
Since $G$ is connected, $Z_{A, D}(G) = 0$ unless $G$ is bipartite.
If $G$ is bipartite with a vertex bipartization $V_1 \cup V_2$,
then we only need to consider  maps $\xi \colon G \to [m]$ such that
either $\xi(V_1) \subseteq [\ell],\, \xi(V_2) \subseteq [\ell + 1, m]$
or $\xi(V_1) \subseteq [\ell + 1, m],\, \xi(V_2) \subseteq [\ell]$,
with all other maps contribute zero to $Z_{A, D}(G)$.
Then
\begin{align*}
Z_{A, D}(G) &= 
\left(
\prod_{u \in V_1} \sum_{i = 1}^\ell \alpha_i x_i^{\deg(u)}
\right)
\left(
 \prod_{v \in V_2} \sum_{j = \ell+1}^{m} \alpha_{j} y_j^{\deg(v)}
\right)
\\
&+ \left(
\prod_{u \in V_1} \sum_{j = \ell+1}^{m} \alpha_{j} y_j^{\deg(u)}
\right) 
\left(
 \prod_{v \in V_2} \sum_{i = 1}^\ell \alpha_i x_i^{\deg(v)}
\right).
\end{align*}
\end{enumerate}

\section{Two technical lemmas}\label{sec:Two-technical-lemmas}



We need two technical lemmas. The following lemma is from \cite{Dyer-Greenhill-2000}
(Lemma 3.6); for the convenience of readers we give a proof here.
%
\begin{lemma}\label{lem:ADA-pairwise-lin-ind}
Let $A$ and $D$ be $m \times m$ matrices,
where $A$ is real symmetric
with all columns nonzero and pairwise linearly independent,
and $D$ is positive diagonal.
Then all columns of $A D A$ are nonzero and pairwise linearly independent.
\end{lemma}
\begin{proof}
The case $m = 1$ is trivial. Assume $m \ge 2$. 
Let $D = \diag(\alpha_i)_{i = 1}^m$, and  $\Pi =  \diag(\sqrt{\alpha_i})_{i = 1}^m$.
Then $\Pi^2 = D$.
We have  $A D A = Q^T Q$, where $Q = \Pi A$.
Let $q_i$ denote the $i$th column of $Q$.
Then $Q$ has pairwise linearly independent columns.
By the Cauchy-Schwartz inequality,
\[
q_i^T q_j < \left( (q_i^T q_i) (q_j^T q_j) \right)^{1 / 2},
\]
whenever $i \ne j$. 
Then for any $1 \le i < j \le m$, the $i$th and $j$th columns of $A D A$
contain a submatrix
\[
\begin{bmatrix}
q_i^T q_i & q_i^T q_j \\
q_i^T q_j & q_j^T q_j
\end{bmatrix},
\]
so they are linearly independent. 
\end{proof}
The following is also adapted from~\cite{Dyer-Greenhill-2000} (Theorem 3.1). 
\begin{lemma}\label{lem:ADA-nondeg-thick}
Let $A$ and $D$ be $m \times m$ matrices,
where $A$ is real symmetric
with all columns nonzero and pairwise linearly independent, 
and $D$ is positive diagonal.
Then for all sufficiently large positive integers 
$p$, the matrix $B = (A D A)^{\odot p}$
corresponding to the edge gadget in Figure~\ref{fig:ADA-to-p-gadget-advanced} is nondegenerate.
\end{lemma}
\begin{proof}
 
If $m = 1$, then any $p \ge 1$ works. Let $m \ge 2$.
Following the proof of Lemma~\ref{lem:ADA-pairwise-lin-ind},
we have
$q_i^T q_j < \sqrt{(q_i^T q_i) (q_j^T q_j)}$,
 for all $1 \le i < j \le m$. Let
\[
\gamma = \max_{1 \le i < j \le m} \frac{q_i^T q_j }{\sqrt{(q_i^T q_i) (q_j^T q_j)}}  < 1.
\]
%
Let $A' = A D A = Q^T Q$ so $A'_{i j} = q_i^T q_j$.
Then $A'_{i j} \le \gamma \sqrt{A'_{i i} A'_{j j}}$ for all $i \ne j$.
Consider the determinant of $A'$.
Each term of $\det(A')$ has the form
\[
\pm \prod_{i = 1}^m A'_{i \sigma(i)},
\]
where $\sigma$ is a permutation of $[m]$.
Denote $t(\sigma) = |\{ i \mid \sigma(i) \ne i \}|$. Then
\[
\prod_{i = 1}^m A'_{i \sigma(i)} \le \gamma^{t(\sigma)} \prod_{i = 1}^m \sqrt{A'_{i i}} \prod_{i = 1}^m \sqrt{A'_{\sigma(i) \sigma(i)}} = \gamma^{t(\sigma)} \prod_{i = 1}^m A'_{i i}.
\]
Consider the  $p$-thickening of $A'$ for $p \ge 1$.
Each term of $\det\left((A')^{\odot p}\right)$ has the form 
$\pm \prod_{i = 1}^m A'^{~p}_{i \sigma(i)}$
for some permutation $\sigma$ of $[m]$.
Now
\[
| \{ \sigma \mid t(\sigma) = j \} | \le \binom{m}{j} j! \le m^j,
\]
for $0 \le j \le m$.
By separating out the identity permutation
and all other terms, for $p \ge \lfloor \ln(2 m) / \ln(1 / \gamma) \rfloor + 1$,
we have $2 m \gamma^p < 1$, and  
\begin{align*}
\det\left((A')^{\odot p}\right) &\ge \left( \prod_{i = 1}^m A'_{i i} \right)^p - \left( \prod_{i = 1}^m A'_{i i} \right)^p \sum_{j = 1}^m m^j \gamma^{p j} \\
&\ge \left( \prod_{i = 1}^m A'_{i i} \right)^p \left( 1 - \frac{m \gamma^p}{1 - m \gamma^p} \right) = \left( \prod_{i = 1}^m A'_{i i} \right)^p \left( \frac{1 - 2 m \gamma^p}{1 - m \gamma^p} \right) > 0.
\end{align*} \end{proof}

\begin{figure}
\centering
\includegraphics{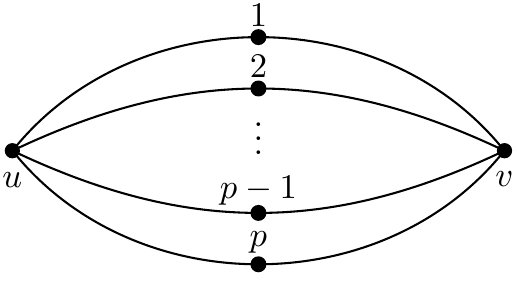} 
\caption{\label{fig:ADA-to-p-gadget-advanced}The edge gadget $S_2 T_p e,\, e = (u, v)$ with the edge weight matrix $(A D A)^{\odot p}$.} 
\end{figure}


\section{Hardness for $Z_A(\cdot)$ on simple graphs for real symmetric $A$}\label{sec:Goldberg-et-al-2010-dichotomy}

There is a more direct approach to prove 
the \#P-hardness part of the Bulatov-Grohe dichotomy
(Theorem~\ref{thm:Bulatov-Grohe}) for  simple graphs.
Although this method does not handle degree-boundedness,
we can apply it more generally to the problem $\EVAL(A, D)$
when the matrix $A$ is real symmetric
and $D$ is positive diagonal.
In particular, we will prove 
the \#P-hardness part of the dichotomy for counting GH
by Goldberg et al.~\cite{Goldberg-et-al-2010} (the problem $\EVAL(A)$ 
without vertex weights, where $A$ is a real symmetric matrix)
for  simple graphs.

We first prove the following theorem.
\begin{theorem}\label{thm:EVAL-simp-interp}
Let $A$ and $D$ be $m \times m$ matrices,
where $A$ is real symmetric and $D$ is positive diagonal.
Then $\EVAL(A, D) \le_{\mathrm T}^{\mathrm P} \EVAL_{\simp}(A, D)$.
\end{theorem}

\begin{proof}
We may assume $A$ is not identically $0$, for otherwise the problem is
trivial.
Let $G = (V, E)$ be an input graph to the problem $\EVAL(A, D)$.
For any $n \ge 1$, let $G_n = S_n^{(F)}(G)$ where $F \subseteq E$ is the
subset consisting of the edges of $G$ each of which is parallel 
to at least one other edge. 
In other words, we obtain $G_n$ by replacing every parallel edge $e$
by its $n$-stretching $S_n e$.
We will refer to these as paths of length $n$ in $G_n$.
Note that $G_1 = G$.
Moreover, for every $n \ge 2$, the graph $G_n$ is simple and loopless,
and has polynomial size in the size of $G$ and $n$.

A path of length $n \ge 1$ has the edge weight matrix
\[
M^{(n)}=
\underbrace{A D A \ldots A D A}_{D \text{ appears } n - 1 ~\ge~ 0 \text{ times}} = A (D A)^{n - 1} = D^{-1 / 2} (D^{1 / 2} A D^{1 / 2})^n D^{-1 / 2}.
\]
Here $D^{1 / 2}$ is a diagonal matrix
with the positive square roots
of the corresponding entries
of $D$ on the main diagonal,
and $D^{-1 / 2}$ is its inverse.

Since $A$ is real symmetric and $D$ is positive diagonal,
the matrix $\widetilde A = D^{1 / 2} A D^{1 / 2}$ is real symmetric.
Then $\widetilde A$ is orthogonally diagonalizable over $\mathbb{R}$, i.e.,
there exist a real orthogonal matrix $S$
and a real diagonal matrix $J = (\lambda_i)_{i = 1}^m$ such that $\widetilde A = S^T J S$.
If $A$ has rank $r$, then $1 \le r \le m$,
and we may assume
that $\lambda_i \ne 0$ for $1 \le i \le r$
and $\lambda_i = 0$ for $i > r$.
%
 
We have ${\widetilde A}^n = S^T J^n S$, so the edge weight matrix for a path of length $n \ge 1$ can be written as 
\[
M^{(n)} 
= D^{-1 / 2} {\widetilde A}^n D^{-1 / 2}
= D^{-1 / 2} S^T J^n S D^{-1 / 2}.
\]

We can write $M_{i j}^{(n)} = \sum_{\ell = 1}^r a_{i j \ell} \lambda_\ell^n$
by a formal  expansion,
for every $n \ge 1$ and some real $a_{i j \ell}$'s
that are dependent on $D$ and $S$, but independent of $n$ and $\lambda_\ell$,
where $1 \le i, j \le m$ and $1 \le \ell \le r$. 
%
%
By the formal expansion of the symmetric matrix $M^{(n)}$ above,
we have $a_{i j \ell} = a_{j i \ell}$.
Let $t = |F|$, which is the number of edges in $G$
subject to the stretching operator $S_n$ to form $G_n$.

In the evaluation of the partition function $Z_{A, D}(G_n)$,
we stratify the vertex assignments in $G_n$ as follows.
Denote by  $\kappa = (k_{i j})_{1 \le i \le j \le m}$
a nonnegative tuple with entries indexed by ordered pairs of nonnegative numbers
that satisfy $\sum_{1 \le i \le j \le m} k_{i j} = t$.
Let $\mathcal K$ denote the set of all such possible tuples $\kappa$.
In particular, $|\mathcal K| = \binom{t + m (m + 1) / 2 - 1}{m (m + 1) / 2 - 1}$.
For a fixed $m$, this is a polynomial in $t$, and thus a
polynomial in the size of $G$.
Let $c_\kappa$ 
be the sum over all assignments of all vertex and edge weight products
in $Z_{A, D}(G_n)$,
 except the contributions by the paths of length $n$ formed by
stretching parallel edges in $G$,
such that the endpoints of precisely $k_{i j}$ constituent paths
 of length $n$
receive the assignments $(i, j)$ (in either order of the end points)
 for every $1 \le i \le j \le m$.
Technically we can call a vertex assignment on $G$
consistent with $\kappa$ (where $\kappa \in \mathcal K$),
if it satisfies the stated property.
Note that the contribution by each such path does not include
the vertex weights of the two end points (but does include all vertex weights
of the internal $n-1$ vertices of the path).
We can write
\[
c_\kappa = \sum_{\substack{\xi \colon V(G) \to [m] \\ \xi \text{ is consistent with } \kappa}} \prod_{w \in V} D_{\xi(w)} \prod_{(u, v) \in E \setminus F} A_{\xi(u), \xi(v)}
\]
for $\kappa \in \mathcal K$.

In particular, the values $c_\kappa$ are independent of $n$.
Thus for some polynomially many values $c_\kappa$, where $\kappa
\in \mathcal K$, we have
\[
Z_{A, D}(G_n) = \sum_{\kappa \in \mathcal K} c_\kappa \prod_{1 \le i \le j \le m} (M_{i j}^{(n)})^{k_{i j}}
= \sum_{\kappa \in \mathcal K} c_\kappa \prod_{1 \le i \le j \le m} (\sum_{\ell = 1}^r a_{i j \ell} \lambda_\ell^n)^{k_{i j}}. 
\]
Expanding out the last sum and rearranging the terms, for some
values $b_{i_1, \ldots, i_r}$  independent of $n$, we get
\begin{equation}\label{interpolation-lin-sys-sec4}
Z_{A, D}(G_n)
= \sum_{\substack{i_1 + \ldots + i_r = t \\ i_1, \ldots, i_r \ge 0}} b_{i_1, \ldots, i_r} ( \prod_{\ell = 1}^r \lambda_\ell^{i_\ell} )^n
\end{equation}
for all $n \ge 1$.

This can be viewed as a linear system with the unknowns $b_{i_1, \ldots, i_r}$ with the rows indexed by $n$. 
The number of unknowns is  $\binom{t + r - 1}{r - 1}$
which is polynomial in the size of the input graph $G$, since $r \le m$ is a constant.
The values $\prod_{\ell = 1}^r \lambda_\ell^{i_\ell}$
can all be computed in polynomial time.

We show how to compute the value $Z_{A, D}(G) = \displaystyle \sum_{\scriptsize \substack{i_1 + \ldots + i_r = t \\ i_1, \ldots, i_r \ge 0}} b_{i_1, \ldots, i_r} \prod_{\ell = 1}^r \lambda_\ell^{i_\ell}$,
from the values $Z_{A, D}(G_n)$ where $n \ge 2$ in polynomial time
(recall that $G_n$ is simple and loopless for $n \ge 2$).
The coefficient matrix of the linear system (\ref{interpolation-lin-sys-sec4})
 is a Vandermonde matrix.
However, it might not be of full rank
because the coefficients $\prod_{\ell = 1}^r \lambda_\ell^{i_\ell}$
do not have to be pairwise distinct, and therefore
it can have repeating columns.
Nevertheless, when there are two repeating columns we replace
the corresponding unknowns $b_{i_1, \ldots, i_r}$ and $b_{i'_1, \ldots, i'_r}$
with their sum as a new variable; we repeat this replacement procedure until
there are no repeating columns.
Since all $\lambda_\ell \ne 0$, for $1 \le \ell \le r$,
after the replacement,
 we have a Vandermonde system of full rank.
Therefore we can solve this modified linear system
in polynomial time.
This allows us to obtain the value $Z_{A, D}(G) = \displaystyle \sum_{\scriptsize \substack{i_1 + \ldots + i_r = t \\ i_1, \ldots, i_r \ge 0}} b_{i_1, \ldots, i_r} \prod_{\ell = 1}^r \lambda_\ell^{i_\ell}$, which also has exactly the same pattern of repeating multipliers
$\prod_{\ell = 1}^r \lambda_\ell^{i_\ell}$. 

We have shown how to compute the value $Z_{A, D}(G)$ in polynomial time
by querying  the oracle $\EVAL(A, D)$ 
on polynomially many instances $G_n$, for $n \ge 2$.
It follows that $\EVAL(A, D) \le_{\mathrm T}^{\mathrm P} \EVAL_{\simp}(A, D)$.
\end{proof}


%
We are ready to prove the \#P-hardness part of the dichotomy
by Goldberg et al.~\cite{Goldberg-et-al-2010} (Theorem 1.1) for simple graphs.
Let $A$ be a real symmetric $m \times m$ matrix.
Assuming that $A$ does not satisfy
the tractability conditions of the dichotomy theorem of Goldberg et al.,
the problem $\EVAL(A)$ is \#P-hard.
By Theorem~\ref{thm:EVAL-simp-interp} (with $D = I_m$), $\EVAL(A) \le_{\mathrm T}^{\mathrm P} \EVAL_{\simp}(A)$.
It follows that $\EVAL_{\simp}(A)$ is \#P-hard.

Hence 
the dichotomy theorem by Goldberg et al. 
can improve to apply to simple graphs.
\begin{theorem}
Let $A$ be a real symmetric matrix.
Then either $\EVAL(A)$ is in polynomial time or $\EVAL_{\simp}(A)$ is \#P-hard
(a fortiori, $\EVAL(A)$ is \#P-hard).

Moreover, there is a polynomial time algorithm that,
given the matrix $A$, decides which case of the dichotomy it is. 
\end{theorem}


\begin{remark}
The interpolation argument in Theorem~\ref{thm:EVAL-simp-interp}
works even if $G$ is a multigraph possibly with multiple loops at any vertex
in the following sense.
In Definition~\ref{def:EVAL(A,D)},
we treat the loops of $G$ as edges.
We think of them as mapped to the entries $A_{i i}$
in the evaluation of the partition function $Z_{A, D}$.
However, we need to slightly change the way we define the graphs $G_n$.
In addition to $n$-stretching the parallel edges of $G$,
we also need to $n$-stretch each loop of $G$ (i.e., replacing a loop by
a closed path of length $n$).
Now $F$ is the set of parallel edges \emph{and} loops in $G$.
This way each $G_n = S_n^{(F)}(G)$ for $n \ge 2$ is simple and loopless.
The rest of the proof goes through.
In other words, the statement of Theorem~\ref{thm:EVAL-simp-interp}
extends to a reduction from
the $\EVAL(A, D)$ problem  that allows input $G$ to have multiloops,
to the standard problem $\EVAL_{\simp}(A, D)$ not allowing loops.
\end{remark}

\bibliography{references}

\begin{thebibliography}{10}

\bibitem{Barvinok-book}
A.~I. Barvinok.
\newblock {\em Combinatorics and Complexity of Partition Functions}, volume~30
  of {\em Algorithms and combinatorics}.
\newblock Springer, 2017.

\bibitem{Barvinok-Soberon-2017}
A.~I. Barvinok and P.~Sober\'on.
\newblock Computing the partition function for graph homomorphisms.
\newblock {\em Combinatorica}, 37(4):633--650, 2017.

\bibitem{baxter-6-8}
R.~J. Baxter.
\newblock The six and eight-vertex models revisited.
\newblock {\em Journal of Statistical Physics}, 116(1):43--66, 2004.

\bibitem{Bulatov-Grohe-2005}
A.~Bulatov and M.~Grohe.
\newblock The complexity of partition functions.
\newblock {\em Theor. Comput. Sci.}, 348(2-3):148--186, 2005.
\newblock A preliminary version appeared in ICALP 2004: 294--306.

\bibitem{cai-chen-book}
J.-Y. Cai and X.~Chen.
\newblock {\em Complexity Dichotomies for Counting Problems}, volume 1: Boolean
  Domain.
\newblock Cambridge University Press, 2017.
\newblock \href {https://doi.org/10.1017/9781107477063}
  {\path{doi:10.1017/9781107477063}}.

\bibitem{Cai-Chen-2019}
J.-Y. Cai and X.~Chen.
\newblock A decidable dichotomy theorem on directed graph homomorphisms with
  non-negative weights.
\newblock {\em Computational Complexity}, 28(3):345--408, 2019.

\bibitem{Cai-Chen-Lu-2013}
J.-Y. Cai, X.~Chen, and P.~Lu.
\newblock Graph homomorphisms with complex values: A dichotomy theorem.
\newblock {\em SIAM J. Comput.}, 42(3):924--1029, 2013.

\bibitem{DyerFJ02}
M.~E. Dyer, A.~M. Frieze, and M.~Jerrum.
\newblock On counting independent sets in sparse graphs.
\newblock {\em {SIAM} J. Comput.}, 31(5):1527--1541, 2002.

\bibitem{Dyer-Goldberg-Paterson-2007}
M.~E. Dyer, L.~A. Goldberg, and M.~Paterson.
\newblock On counting homomorphisms to directed acyclic graphs.
\newblock {\em J. ACM}, 54(6):27, 2007.

\bibitem{Dyer-Greenhill-2000}
M.~E. Dyer and C.~S. Greenhill.
\newblock The complexity of counting graph homomorphisms.
\newblock {\em Random Struct. Algorithms}, 17(3-4):260--289, 2000.
\newblock A preliminary version appeared in SODA 2000: 246--255.

\bibitem{Dyer-Greenhill-corrig-2004}
M.~E. Dyer and C.~S. Greenhill.
\newblock Corrigendum: The complexity of counting graph homomorphisms.
\newblock {\em Random Struct. Algorithms}, 25(3):346--352, 2004.

\bibitem{Goldberg-et-al-2010}
L.~A. Goldberg, M.~Grohe, M.~Jerrum, and M.~Thurley.
\newblock A complexity dichotomy for partition functions with mixed signs.
\newblock {\em SIAM J. Comput.}, 39(7):3336--3402, 2010.

\bibitem{Grohe-Thurley-2011}
M.~Grohe and M.~Thurley.
\newblock Counting homomorphisms and partition functions.
\newblock In M.~Grohe and J.~Makowsky, editors, {\em Model Theoretic Methods in
  Finite Combinatorics}, volume 558 of {\em Contemporary Mathematics}, pages
  243--292. American Mathematical Society, 2011.

\bibitem{GH-book}
P.~Hell and J.~Ne\v{s}et\v{r}il.
\newblock {\em Graphs and homomorphisms}, volume~28 of {\em Oxford lecture
  series in mathematics and its applications}.
\newblock Oxford University Press, 2004.

\bibitem{HelmuthPR19}
T.~Helmuth, W.~Perkins, and G.~Regts.
\newblock Algorithmic {P}irogov-{S}inai theory.
\newblock In {\em Proceedings of the 51st Annual ACM SIGACT Symposium on Theory
  of Computing (STOC)}, pages 1009--1020, 2019.

\bibitem{Lenstra-1992}
H.~W.~Lenstra Jr.
\newblock Algorithms in algebraic number theory.
\newblock {\em Bull. Amer. Math. Soc.}, 26(2):211--244, 1992.

\bibitem{LiLY13}
L.~Li, P.~Lu, and Y.~Yin.
\newblock Correlation decay up to uniqueness in spin systems.
\newblock In {\em Proceedings of the 24th Annual ACM-SIAM Symposium on Discrete
  Algorithms (SODA)}, pages 67--84, 2013.

\bibitem{Lovasz-1967}
L.~Lov\'asz.
\newblock Operations with structures.
\newblock {\em Acta Math. Hungar.}, 18(3-4):321--328, 1967.

\bibitem{Peters-Regts-2018}
H.~Peters and G.~Regts.
\newblock Location of zeros for the partition function of the {I}sing model on
  bounded degree graphs.
\newblock {\em arXiv:1810.01699}, 2018.
\newblock URL: \url{https://arxiv.org/abs/1810.01699}.

\bibitem{SinclairST12}
A.~Sinclair, P.~Srivastava, and M.~Thurley.
\newblock Approximation algorithms for two-state anti-ferromagnetic spin
  systems on bounded degree graphs.
\newblock In {\em Proceedings of the 23rd Annual ACM-SIAM Symposium on Discrete
  Algorithms (SODA)}, pages 941--953, 2012.

\bibitem{Sly10}
A.~Sly.
\newblock Computational transition at the uniqueness threshold.
\newblock In {\em Proceedings of the 51st Annual IEEE Symposium on Foundations
  of Computer Science (FOCS)}, pages 287--296, 2010.

\bibitem{Thurley-2009}
M.~Thurley.
\newblock {\em The Complexity of Partition Functions}.
\newblock PhD thesis, Humboldt Universit\"at zu Berlin, 2009.

\bibitem{Thurley-2010}
M.~Thurley.
\newblock The complexity of partition functions on {H}ermitian matrices.
\newblock {\em arXiv:1004.0992}, 2010.
\newblock URL: \url{https://arxiv.org/abs/1004.0992}.

\bibitem{Valiant}
L.~G. Valiant.
\newblock The complexity of enumeration and reliability problems.
\newblock {\em {SIAM} J. Comput.}, 8(3):410--421, 1979.

\bibitem{Weitz06}
D.~Weitz.
\newblock Counting independent sets up to the tree threshold.
\newblock In {\em Proceedings of the 38th Annual ACM SIGACT Symposium on Theory
  of Computing (STOC)}, pages 140--149, 2006.

\bibitem{MingjiXia}
M.~Xia.
\newblock Holographic reduction: A domain changed application and its partial
  converse theorems.
\newblock {\em Int. J. Software and Informatics}, 5(4):567--577, 2011.
\newblock A preliminary version appeared in ICALP 2010 (1): 246--255.

\end{thebibliography}

\end{document}